\documentclass[twoside,12pt]{article}
\usepackage{indentfirst}
\usepackage{bm}
\usepackage{cite}
\usepackage{graphicx}
\usepackage{epsfig}
\usepackage{amsmath}
\usepackage{amsfonts}
\usepackage{amssymb}
\usepackage{amsthm}
\usepackage{latexsym}
\usepackage{amsmath}
\usepackage{pgfplots}

\usepackage{booktabs}
\usepackage{caption}
\newtheorem{thm}{Theorem}

\usepackage{epsfig}
\usepackage{amsmath}
\usepackage{epstopdf}
\usepackage{pgf,fancyhdr}
\usepackage{float}
\usepackage{tabularx}
\newcolumntype{Y}{>{\centering\arraybackslash}X}
\usepackage{makecell}
\newcolumntype{V}{!{\vrule width 2pt}X}

\usepackage{soul}
\usepackage{color,xcolor}
\soulregister\cite7

\begin{document}

\begin{center}
\Large\bf Geometric genuine $N$-partite entanglement measure for arbitrary dimensions
\end{center}

\begin{center}
\rm  Hui Zhao,$^{*1}$ Pan-Wen Ma,$^{1}$ \ Shao-Ming Fei$^{2}$ and \ Zhi-Xi Wang $^{2}$
\end{center}

\begin{center}
\begin{footnotesize} \sl
$^1$ School of Mathematics, Statistics and Mechanics, Beijing University of Technology, Beijing 100124, China

$^2$ School of Mathematical Sciences,  Capital Normal University,  Beijing 100048, China

\end{footnotesize}
\end{center}
\footnotetext{\\{*}Corresponding author, zhaohui@bjut.edu.cn}
\vspace*{2mm}

\begin{center}
\begin{minipage}{15.5cm}
\parindent 20pt\footnotesize
We present proper genuine multipartite entanglement (GME) measures for arbitrary multipartite and dimensional systems. By using the volume of concurrence regular polygonal pyramid we first derive the GME measure of four-partite quantum systems. From our measure it is verified that the GHZ state is more entangled than the W state. Then we study the GME measure for multipartite quantum states in arbitrary dimensions. A well defined GME measure is constructed based on the volume of the concurrence regular polygonal pyramid. Detailed example shows that our measure can characterize better the genuine multipartite entanglements.
\end{minipage}
\end{center}

\begin{center}
\begin{minipage}{15.5cm}
\begin{minipage}[t]{2.3cm}{\bf Keywords:}\end{minipage}
\begin{minipage}[t]{13.1cm}
Genuine multipartite entanglement, genuine multipartite entanglement measure, concurrence
\end{minipage}\par\vglue8pt
\end{minipage}
\end{center}

\section{Introduction}
Quantum entanglement \cite{ref1,ref2} is one of the novel phenomena of quantum mechanics that distinguishes the quantum world and the classical one. It plays a central role in advanced quantum technologies such as quantum teleportation \cite{ref3,ref4}, quantum key distribution \cite{ref5}, superdense coding \cite{ref6} and quantum computing \cite{ref7,ref8}. The theory of quantum entanglement has attracted much attention. For a bipartite system, entanglement is quantified by various measures such as concurrence \cite{ref9}, negativity \cite{ref10} and entanglement formation \cite{ref11}. A particular important class of multipartite entanglement is the genuine multipartite entanglement (GME) \cite{ref12,ref13, ref14}. A multipartite state is genuinely entangled when it can't be expressed as the convex combination of biseparable states (separable under some bipartitions). A key issue is to give a suitable measure of GME for quantifying the genuine multipartite entanglement.

A well defined GME measure $E$ has to satisfy the following conditions. (i)For all product and biseparable states, the measure must be zero. (ii) It is positive for all non-biseparable states. (iii) It is non-increasing under local operations and classical communications (LOCC), { i.e., $E(\Lambda_{\text{LOCC}}(\rho))\leq E(\rho)$}. A GME measure of any product or biseparable states takes value zero [15,16]. Namely, when the value of the GME measure is $0$, the state is separable under one kind of partition, but the other subsystems may still remain entanglement. The vanishing GME measure does not necessarily imply the fully separability. This is different from the condition of an entanglement measure [17]: it is non-zero if and only if the state is entangled.
Although the experimental observation of multipartite entanglement has been successfully realized in \cite{ref18,ref19}, the quantification of multipartite entanglement is still far from being satisfied.

There are three known GME measures for three-qubit systems. The genuine multipartite concurrence presented by Ma et al. \cite{ref20} is exactly the minimum concurrence between each single qubit and the remaining ones. Based on the distance between a given state and its closest biseparable states, the authors provided the generalized geometric measure in \cite{ref21}. In \cite{ref22}, the authors presented a measure which is given by the average of 3-tangle. Very recently, Xie and Eberle \cite{ref15} introduced a measure to quantify the GME of three-qubit systems, which has a simple form and elegant geometric interpretation. Whereafter, in \cite{ref16}, the authors proposed an improved GME measure by using the geometric mean area of these concurrence triangles. An approach of constituting GME measure using the area of the triangle and the superficial area of the tetrahedron was given in \cite{ref23}, while whether such measures are non-increasing under LOCC remains unknown. Based on the concurrence of nine different classes of four-qubit states, proper genuine four-qubit entanglement measure is presented in \cite{ref24} by using the volume of the concurrence tetrahedron. Nevertheless, extending this geometric measurement to high dimensional systems is not as simple as it seems.

In this paper, we study the GME measures by using concurrence regular polygonal pyramid. The paper is organized as follows. In Section 2, we first review some basic concepts regarding concurrence and construct concurrence rectangular pyramid to derive GME measures for four-partite quantum systems. By detailed example, we show that our criteria are more efficient than the existing one. In Section 3, we present a well defined GME measure by using the volume of concurrence regular polygonal pyramid for multipartite quantum states in arbitrary dimensions. Conclusions are given in Section 4.

\section{GME measure for four-partite pure states}
We first consider the GME measure for four-partite systems. Let $\mathcal {H}_f$ denote a $d_{f}$-dimensional Hilbert vector space. For a bipartite pure state $\rho=|\psi_{12}\rangle\langle\psi_{12}|$ in a finite-dimensional Hilbert space $ \mathcal {H}_1\otimes  \mathcal {H}_2=\mathbb{C}^{d_1}\otimes\mathbb{C}^{d_2}$, the concurrence is given by \cite{ref9},
\begin{equation}
C(|\psi_{12}\rangle)=\sqrt{2[1-Tr{(\rho_1^2)}]},
\end{equation}
where $\rho_1=Tr_2(\rho)$ is the reduced density matrix of the first subsystem.

We have two types of entanglements under two different bipartitions, $C_{i|jkl}(|\psi\rangle)$ with respect to one-to-three bipartition and $C_{ij|kl}(|\psi\rangle)$ with respect to two-to-two bipartition, where $i\neq j\neq k\neq l\in\{1,2,3,4\}$. Denote
\begin{equation}
a=\sqrt[4]{C_{1|234}(|\psi\rangle)C_{2|134}(|\psi\rangle)
C_{3|124}(|\psi\rangle)C_{4|123}(|\psi\rangle)}
\end{equation}
and
\begin{equation}
h=\sqrt[3]{C_{12|34}(|\psi\rangle)C_{13|24}(|\psi\rangle)C_{14|23}(|\psi\rangle)}.
\end{equation}
We obtain a rectangular pyramid with $a$ as the length of the bottom edges and $h$ as the height, see Figure 1. We call it concurrence rectangular pyramid. The volume $\frac{1}{3}a^2*h$ of the concurrence rectangular pyramid has a simple form and an elegant geometric interpretation of entanglement. For both product and biseparable states, the volume is zero and thus the condition (i) of GME measure is satisfied. {This shows a correlation between the four-partite entanglement and the concurrence rectangular pyramid. The fully separable, bi-separable and non-biseparable states are associated with the concrete vertices, edges or facets and rectangular pyramid, respectively, providing a geometric identification of entanglement.} In Figure 2 we list the relations between the entanglement of four-partite pure states and the corresponding concurrence rectangular pyramid.
\begin{figure}[ht]
  \centering{\includegraphics[width=0.5\textwidth]{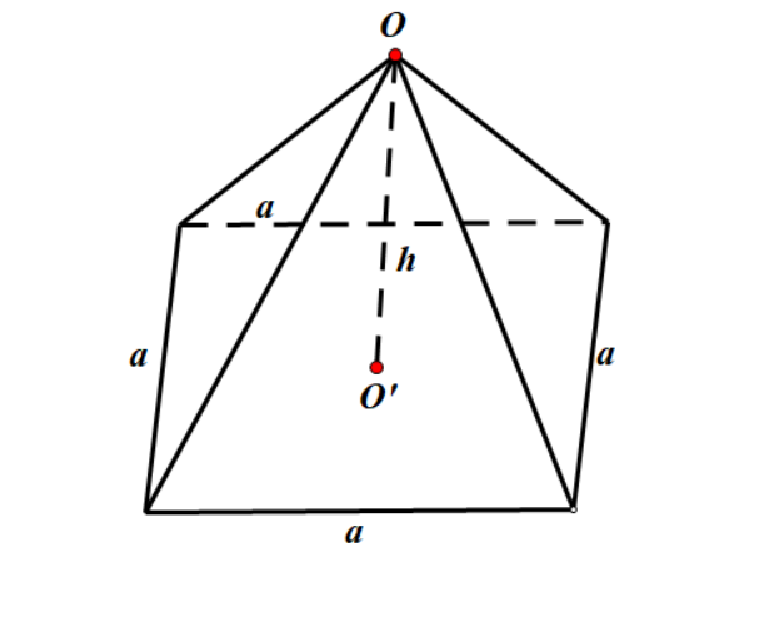}}\\
  \caption{Schematic diagram of concurrent rectangular pyramid. The base of the concurrent rectangular pyramid is a square with side length $a$, $O'$ is the center of the square. $O$ is the vertex of the concurrent rectangular pyramid. $OO'$ is the height of the concurrent rectangular pyramid perpendicular to the base with length $h$.}\label{Figure.1}
\end{figure}

Using the volume of the concurrence rectangular pyramid to define a GME measure for four-partite pure states, we have the following theorem.

\begin{figure}[ht]
  \centering{\includegraphics[width=0.98\textwidth]{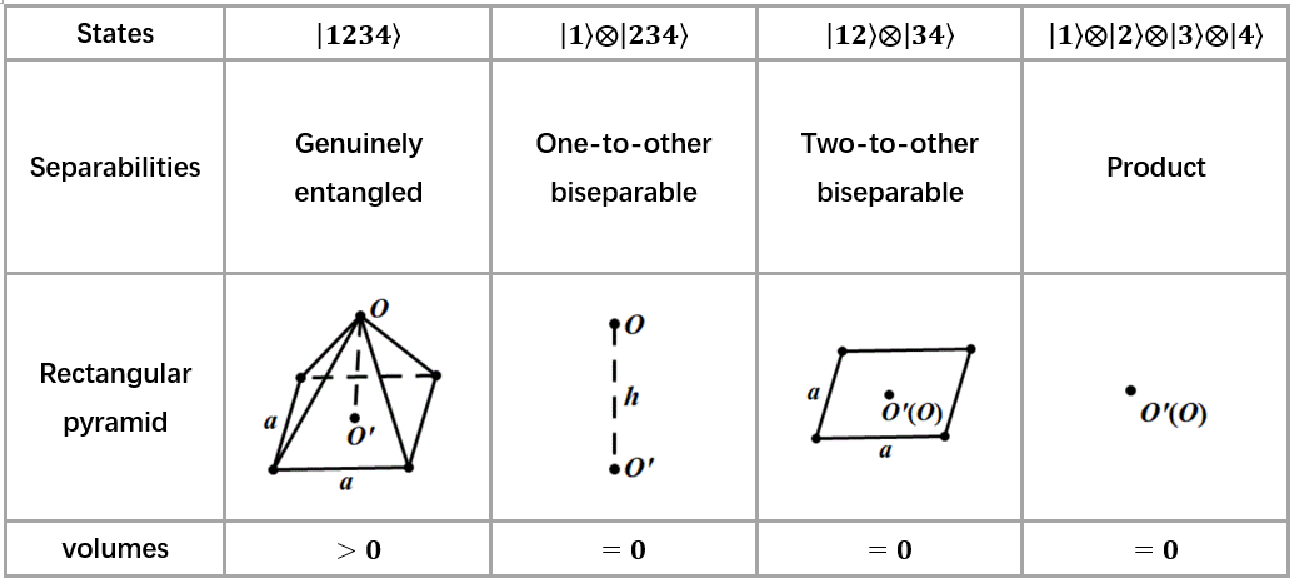}}\\
  \caption{Schematic diagram of entanglement and concurrence rectangular pyramid.}\label{Figure.2}
\end{figure}

\begin{thm}
For any four-partite pure state $|\psi\rangle\in \mathcal {H}_{1}^{d_1}\otimes \mathcal {H}_{2}^{d_2}\otimes \mathcal {H}_{3}^{d_3}\otimes \mathcal {H}_{4}^{d_4}$, the volume of the concurrence rectangular pyramid defines a GME measure,
\begin{equation}
V_{1234}(|\psi\rangle)=\frac{1}{3}{a}^{2}h,
\end{equation}
where $a$ is the length of bottom edges and $h$ is the height of the concurrence rectangular pyramid.
\end{thm}

\begin{proof}
We first prove that $|\psi\rangle$ is GME iff $V_{1234}(|\psi\rangle)>0$. On one hand, if $V_{1234}(|\psi\rangle)>0$, then $a>0$ and $h>0$, that is to say, $C_{i|jkl}(|\psi\rangle)$ and $C_{ij|kl}(|\psi\rangle)$ are all positive. Hence $|\psi\rangle$ is a genuine multipartite entangled state. On the other hand, if $V_{1234}(|\psi\rangle)=0$, the volume of the concurrence rectangular pyramid is zero iff either the length of bottom edge or the height of the pyramid or both are $0$. We obtain that at least one of $\{C_{1|234}(|\psi\rangle), C_{2|134}(|\psi\rangle), C_{3|124}(|\psi\rangle), C_{4|123}(|\psi\rangle),\\
C_{12|34}(|\psi\rangle), C_{13|24}(|\psi\rangle), C_{14|23}(|\psi\rangle)\}$ is $0$. Without loss of generality, suppose $C_{13|24}(|\psi\rangle)=0$, namely, $|\psi\rangle=\rho_{13}\otimes\rho_{24}$. Therefore $|\psi\rangle$ is not GME.

We next prove that $V_{1234}(|\psi\rangle)$ can't increase under LOCC, i.e., $V_{1234}(\Lambda|\psi\rangle)\leq V_{1234}(|\psi\rangle)$ for any LOCC map $\Lambda$. As concurrence $C$ is non-increasing under LOCC, we only need to prove that $V_{1234}(|\psi\rangle)$ is an increasing function of ${C_{i|jkl}(|\psi\rangle), C_{ij|kl}(|\psi\rangle)}$. Actually, by direct calculation we have
$$
\frac{\partial{V_{1234}(|\psi\rangle)}}{\partial{C_{i|jkl}(|\psi\rangle)}}=
\frac{1}{6} h {(C_{i|jkl}(|\psi\rangle))}^{-\frac{1}{2}}{(C_{j|ikl}(|\psi\rangle)
C_{k|ijl}(|\psi\rangle)C_{l|ijk}(|\psi\rangle))}^{\frac{1}{2}}\geq0
$$
and
$$
\frac{\partial{V_{1234}(|\psi\rangle)}}{\partial{C_{ij|kl}(|\psi\rangle)}}=
\frac{1}{9} a^2 {(C_{ij|kl}(|\psi\rangle))}^{-\frac{2}{3}}{(C_{ik|jl}(|\psi\rangle)
C_{il|jk}(|\psi\rangle))}^{\frac{1}{3}}\geq0.
$$
Thus the monotonicity of $V_{1234}(|\psi\rangle)$ holds and $V_{1234}(|\psi\rangle)$ is non-increasing under LOCC.
\end{proof}

{\textbf{Remark 1}: Here the monotonicity means that the measure does not increase under any LOCC, i.e., $E(\Lambda_{\text{LOCC}}(\rho))\leq E(\rho)$. The strong monotonicity says that if $\rho$ is transformed into a state $\sigma_j$ with probability $p_j$ under LOCC, the measure is non-increasing on average, namely, $ \sum\limits_{j}p_j E(\sigma_j)\leq E(\rho)$ for the LOCC-ensemble $\{p_j, \sigma_j \}$. In this paper, for a well defined GME measure, we only require the monotonicity, but not the strong monotonicity, see \cite{ref16,ref20,ref24}.}

In \cite{ref15} the authors suggest that a GME measure should satisfy a additional condition (iv) GME of the GHZ state is larger than that of the W state, and a measure satisfying all the conditions (i)-(iv) is called a proper GME measure. Here, for four-qubit pure GHZ state
$|GHZ\rangle= \frac{1}{\sqrt2}{(|0000\rangle+|1111\rangle)}$ and the $W$ state
$|W\rangle= \frac{1}{2}{(|1000\rangle+|0100\rangle+|0010\rangle+|0001\rangle)},$
we obtain $V_{1234}(|GHZ\rangle)=0.3333$ due to that $C_{i|jkl}(|\psi\rangle)$ and $C_{ij|kl}(|\psi\rangle)$ are equal to 1. For the W state, we have $V_{1234}(|W\rangle)=0.1875$. Obviously, the GHZ state is more entangled than the W state. Thus $V_{1234}(|\psi\rangle)$ is a proper GME measure in this sense.

\textbf{Example 1} Consider four-qubit pure states $|\psi_A\rangle$, $|\psi_B\rangle$, $|\psi_C\rangle$ and $|\psi_D\rangle$ given by
\begin{equation*}
\begin{split}
|\psi_A\rangle & =\frac{1}{2}{(|0000\rangle+|1011\rangle+|1101\rangle+|1110\rangle)},\\
|\psi_B\rangle & =\frac{1}{2}{(|0000\rangle+|0101\rangle+|1000\rangle+|1110\rangle)},\\
|\psi_C\rangle & =\frac{1}{\sqrt5}{(|0000\rangle+|1111\rangle+|0011\rangle+|0101\rangle+|0110\rangle)},\\
|\psi_D\rangle & =\frac{1}{\sqrt{4(\frac{5\sqrt{113}}{32}+\frac{51}{32})+3}}(\sqrt{4(\frac{5\sqrt{113}}{32}+\frac{51}{32})}
(|0000\rangle+|0101\rangle+|1010\rangle+|1111\rangle)\\
& +(\mathrm{i}|0001\rangle+|0110\rangle-\mathrm{i}|1011\rangle)),
\end{split}
\end{equation*}
Where $\mathrm{i}=\sqrt{-1}$.

From Theorem 1 we have $V_{1234}(|\psi_A\rangle)=0.3468$,
$V_{1234}(|\psi_B\rangle)=0.2788$, $V_{1234}(|\psi_C\rangle)=0.1487$ and
$V_{1234}(|\psi_D\rangle)=0.3407$.

In \cite{ref20}, for an four-qubit quantum pure state $|\psi\rangle$, the genuine multipartite concurrence is defined to be $C_{GME}(|\psi\rangle)=\min\limits_{\gamma_t\in\gamma}\sqrt{ 2(1-Tr(\rho_{\gamma_t})^2)}$, where $\gamma=\{\gamma_t\} $ labels all the different reduced density matrices of $|\psi\rangle\langle\psi|$. Direct calculation shows that $C_{GME}(|\psi_A\rangle)=C_{GME}(|\psi_B\rangle)=0.8660$ and $C_{GME}(|\psi_C\rangle)=C_{GME}(|\psi_D\rangle)=0.8000$.

Clearly, the measure $C_{GME}$ in \cite{ref20} can't tell the difference between the entanglement of $|\psi_A\rangle$ and $|\psi_B\rangle$, or of $|\psi_C\rangle$ and $|\psi_D\rangle$. This fact is that $C_{GME}$ only depends on the length of the shortest edge which is the same for both states, while our GME measure $V_{1234}(|\psi\rangle)$ can distinguish.

\section{GME measure for multipartite pure states}
We next consider general $N$-partite systems with subsystem $1, 2, \cdots , N$, for arbitrary  $N$-partite pure states $|\Psi\rangle\in \mathcal {H}_{1}^{d_1}\otimes \mathcal {H}_{2}^{d_2}\otimes \cdots \otimes \mathcal {H}_{N}^{d_N}$. For simplicity, we denote the concurrence between the subsystem $i$ and the rest ones as $C_{i|\hat{i}}(|\Psi\rangle)$, where $\hat{i}$ stands for all the subsystems without the $i$th one.

For general $N$-partite pure states $|\Psi\rangle$, there are $\lfloor\frac{N}{2}\rfloor$ different types of bipartitions, where $\lfloor\cdot\rfloor$ stands for rounding down. Define $C_{i}(|\Psi\rangle)=C_{i| \hat{i}}(|\Psi\rangle)$,  $C_{ij}(|\Psi\rangle)=C_{ij| \hat{ij}}(|\Psi\rangle)$, $C_{i\overrightarrow{j_2}}(|\Psi\rangle)=C_{i\overrightarrow{j_2}| \widehat{i\overrightarrow{j_2}}}(|\Psi\rangle), \cdots, C_{i\overrightarrow{j_{\lfloor\frac{N}{2}\rfloor-1}}}(|\Psi\rangle)
=C_{i\overrightarrow{j_{\lfloor\frac{N}{2}\rfloor-1}}| \widehat{i\overrightarrow{j_{\lfloor\frac{N}{2}\rfloor-1}}}}(|\Psi\rangle)$, where $\overrightarrow{j_m}=(jj_2\cdots j_m)$ for $i\neq j\neq\cdots\neq j_{\lfloor\frac{N}{2}\rfloor-1}\in\{1,2,\cdots,N\}$. Let $a$ be the geometric mean of $\{C_{i}(|\Psi\rangle)\}$ and $h$ the geometric mean of $\{C_{ij}(|\Psi\rangle), C_{i\overrightarrow{j_2}}(|\Psi\rangle), \cdots, C_{i\overrightarrow{j_{\lfloor\frac{N}{2}\rfloor-1}}}(|\Psi\rangle)\}$,
respectively,
\begin{equation}
a=\sqrt[N]{\prod\limits_{i=1}^N C_{i}(|\Psi\rangle)},
\end{equation}
\begin{equation}
h=\sqrt[2^{N-1}-N-1]{\prod\limits^{P_N^2} C_{ij}(|\Psi\rangle)\prod\limits^{P_N^3} C_{i\overrightarrow{j_2}}(|\Psi\rangle)\cdots\prod\limits^{P_N^{\lfloor\frac{N}{2}\rfloor-1}} C_{i\overrightarrow{j_{\lfloor\frac{N}{2}\rfloor-2}}}(|\Psi\rangle)\prod\limits^{\frac{P_N^{\lfloor\frac{N}{2}\rfloor}}{2}} C_{i\overrightarrow{j_{\lfloor\frac{N}{2}\rfloor-1}}}(|\Psi\rangle)},
\end{equation}
where $P_n^m=\frac{n!}{m!(n-m)!}$.
We construct a concurrence regular polygonal pyramid with $a$ as the length of the bottom edges and $h$ as the height. The area of the base regular polygon is given by,
\begin{equation}\label{area}
S=\frac{Na^2}{4}\cot(\frac{\pi}{N}).
\end{equation}

\begin{thm}
For any $N$-partite pure state $|\Psi\rangle\in \mathcal {H}_{1}^{d_1}\otimes \mathcal {H}_{2}^{d_2}\otimes \cdots \otimes \mathcal {H}_{N}^{d_N}$, the volume of the concurrence regular polygonal pyramid defines a GME measure,
\begin{equation}\label{thm2}
V(|\Psi\rangle)=\frac{Na^2}{12}\cot{(\frac{\pi}{N})}h
\end{equation}
for $N>3$.
\end{thm}

\begin{proof}
We first prove that $|\Psi\rangle$ is GME iff $V(|\Psi\rangle)>0$. On one hand, note that the cotangent function $y=\cot x>0$ when $0<x<\frac{\pi}{3}$. Hence for $N>3$, if $V(|\Psi\rangle)>0$, then $a>0$ and $h>0$, that is to say, $\{C_{i}(|\Psi\rangle)\}$ and $\{C_{ij}(|\Psi\rangle), C_{i\overrightarrow{j_2}}(|\Psi\rangle), \cdots, C_{i\overrightarrow{j_{\lfloor\frac{N}{2}\rfloor-1}}}(|\Psi\rangle)\}$ are all positive. Hence $|\Psi\rangle$ is a genuine multipartite entangled state. On the other hand, if $V(|\Psi\rangle)=0$, the volume of the concurrence regular polygonal pyramid is zero iff at least one of the length of bottom edge and the height of pyramid is $0$, namely, at least one of $\{C_{i}(|\Psi\rangle),C_{ij}(|\Psi\rangle), C_{i\overrightarrow{j_2}}(|\Psi\rangle), \cdots, C_{i\overrightarrow{j_{\lfloor\frac{N}{2}\rfloor-1}}}(|\Psi\rangle)\}$ is $0$. Without loss of generality, suppose $C_{24}(|\Psi\rangle)$ $=0$. We get $|\Psi\rangle\langle\Psi|=\rho_{24}\otimes\rho_{13\cdots N}$, i.e., $|\Psi\rangle$ is not genuine multipartite entangled.

We next prove that $V(|\Psi\rangle)$ can't increase under LOCC. Similarly, we only need to prove that $V(|\Psi\rangle)$ is an increasing function of $\{C_{i}(|\Psi\rangle),C_{ij}(|\Psi\rangle), C_{i\overrightarrow{j_2}}(|\Psi\rangle), \cdots, C_{i\overrightarrow{j_{\lfloor\frac{N}{2}\rfloor-1}}}(|\Psi\rangle)\}$. Direct calculation shows that
$$
\frac{\partial{V(|\Psi\rangle)}}{\partial{C_{i}(|\Psi\rangle)}}=\frac{h}{6}
\cot{(\frac{\pi}{N})}{C_{i}(|\Psi\rangle)}^{\frac{2}{N}-1}\prod\limits_{k\neq i}^{N-1} {C_{k}(|\Psi\rangle)}^{\frac{2}{N}}\geq0
$$
and
\begin{eqnarray*}
&\frac{\partial{V(|\Psi\rangle)}}{\partial{C_{ij}(|\Psi\rangle)}}&=
\frac{Na^2}{12(2^{N-1}-N-1)}\cot(\frac{\pi}{N}){C_{ij}(|\Psi\rangle)}^{\frac{1}{2^{N-1}-N-1}-1}
\cdot\\
&&\{\prod\limits_{lm\neq ij}^{P_N^2-1}C_{lm}(|\Psi\rangle)\prod\limits^{P_N^3} C_{i\overrightarrow{j_2}}(|\Psi\rangle)\cdots
\prod\limits^{\frac{P_N^{\lfloor\frac{N}{2}\rfloor}}{2}} C_{i\overrightarrow{j_{\lfloor\frac{N}{2}\rfloor-1}}}(|\Psi\rangle)\}^{\frac{1}{2^{N-1}-N-1}}
\geq0
\end{eqnarray*}
for $k\neq i\in\{1,2,\cdots,N\}$ and $l\neq m\in\{1,2,\cdots,N\}$.
Likewise, the nonnegativity of the derivative of $V(|\Psi\rangle)$ with respect to $\{ C_{i\overrightarrow{j_2}}(|\Psi\rangle), \cdots, C_{i\overrightarrow{j_{\lfloor\frac{N}{2}\rfloor-1}}}(|\Psi\rangle)\} $ holds respectively. Thus the monotonicity of $V(|\Psi\rangle)$ holds and $V(|\Psi\rangle)$ is non-increasing under LOCC. Therefore, $V(|\Psi\rangle)$ is a bona fide measure of GME.
\end{proof}

{Our entanglement measure given by (8) has an intuitive and elegant geometric explanation. It is just the volume of the concurrence regular polygonal pyramid, namely, the product of the length of the bottom edges $a$ and the height $h$. The entanglement and separability can be seen in a geometrical way. Characterizing genuine multipartite entanglement in terms of geometric construction is very useful since it helps to better understand GME, and may also be used to understand relevant quantum correlations such as genuine non-locality and steering.}
We remark that when $N=4$, $V(|\Psi\rangle)$ in Theorem 2 reduces to the measure $V_{1234}(|\Psi\rangle)$ in Theorem 1. In Ref.\cite{ref26} the authors indicated that the entanglement can be viewed as a distance between two separated points in a Bloch sphere. It has an elegant geometric interpretation for tripartite symmetric states. Nevertheless, this approach can not be considered as the appropriate entanglement measure to four-partite quantum systems, since in this case the barycenter coincides with the center of the Bloch sphere. We study the GME measures by using concurrence regular polygonal pyramid, which takes into account all different types of bipartitions. A state is GME iff the value of the GME measure is greater than zero.

Moreover, for $N=3$ we denote concurrence regular tetrahedron with $a$ as the length of bottom edges and $h$ as the height. In this case,
$a=\sqrt[3]{C_{1|23}(|\psi\rangle)C_{2|13}(|\psi\rangle)C_{3|12}(|\psi\rangle)}$ and
$h=1$. Then we have the following genuine tripartite entanglement measure,
\begin{equation}\label{9}
V_{123}(|\psi\rangle)=\frac{\sqrt3}{12}a^2.
\end{equation}

In \cite{ref15} the authors used squared concurrence as three edges of a triangle and proposed the following genuine tripartite entanglement measure (triangle measure) for three-qubit states,
\begin{equation}
F_{123}=[\frac{16}{3}Q\prod\limits_{i=1}^3(Q-{{C_{i|\hat{i}}^2(|\psi\rangle)}})]^{\frac{1}{4}},
\end{equation}
where $Q=\frac{1}{2}\sum\limits _{i=1}^3{C_{i|\hat{i}}^2(|\psi\rangle)}$ for $i\in\{1,2,3\}$. {While the authors didn't discuss whether the measure (10) satisfies monotonicity or strong monotonicity, our measure (9) do satisfy the monotonicity, but not the strong monotonicity. It has been shown in\cite{ref25} that the triangle measure (10) satisfies the strong monotonicity if the squared concurrence is replaced with concurrence. Furthermore, the authors in\cite{ref27} demonstrated that this triangle measure satisfies strong monotonicity if one uses entropy instead of concurrence.}

\textbf{Example 2} Consider the following five-qubit pure state $|\phi_{12345}\rangle$,
\begin{equation*}
\begin{split}
|\phi_{12345}\rangle & =\frac{1}{2}{(|00000\rangle+|01010\rangle+|10100\rangle+|11110\rangle)}.
\end{split}
\end{equation*}
By straightforward calculation we have $V(|\phi_{12345}\rangle)=0$. Hence, $\phi_{12345}$ is not a genuine multipartite entangled state. In fact,
$|\phi_{12345}\rangle=\frac{1}{\sqrt2}{(|00\rangle+|11\rangle)_{13}}
\bigotimes\frac{1}{\sqrt2}{(|000\rangle+|110\rangle)_{245}}$.

\section{Conclusion}
We have studied GME measure based on the concurrence regular polygonal pyramid. By using the volume of the concurrence rectangular pyramid, we have constructed GME measures for arbitrary multipartite quantum systems with arbitrary dimensions.
Detailed examples have shown that our results are better in characterizing genuine multipartite entanglements. Our GME measures have explicit and elegant geometric figures.
The results may highlight further investigations on genuine multipartite correlations besides entanglement.

\noindent\textbf{Declaration of competing interest}

The authors declare that they have no known competing financial interests or personal relationships that could have appeared to
influence the work reported in this paper.\\

\noindent\textbf{Data availability}

No data was used for the research described in the article.\\

\noindent\textbf {Acknowledgements}
This work is supported by the National Key R\&D Program of China under Grant No.(2022YFB3806000), National Natural Science Foundation of China under Grants (12272011, 12075159, 12126351 and 12171044), the Academician Innovation Platform of Hainan Province and Simons Foundation.

\end{document}